\documentclass[11pt,a4paper]{article}
\usepackage[top=27mm, bottom=27mm, =22mm, right=22mm]{geometry}
\usepackage{amsthm,amsmath,amssymb}
\usepackage{mathrsfs}
\usepackage{multirow}
\usepackage{amsthm}
\usepackage{microtype}
\usepackage{hyperref}
\usepackage[capitalize]{cleveref}
\usepackage{verbatim}
\usepackage{xcolor}

\newtheorem{theorem}{Theorem}[section]
\newtheorem{lemma}[theorem]{Lemma}
\newtheorem{proposition}[theorem]{Proposition}
\newtheorem{corollary}[theorem]{Corollary}
\newtheorem{conjecture}[theorem]{Conjecture}
\newtheorem{definition}[theorem]{Definition}

\newtheorem{remark}[theorem]{Remark}
\newtheorem{claim}[theorem]{Claim}

\newtheorem{question}[theorem]{Question}

\newtheorem{fact}[theorem]{Fact}

\newcommand{\ma}{\mathcal}

\newcommand{\s}{\subseteq}

\newcommand{\bi}{\binom}
\newcommand{\fr}{\frac}
\newcommand{\lc}{\lceil}
\newcommand{\rc}{\rceil}
\newcommand{\ep}{\epsilon}
\newcommand{\la}{\lambda}
\newcommand{\no}{\noindent}
\newcommand{\Var}{\mathrm{Var}}
\newcommand{\E}{\mathbb{E}}

\begin{document}

\title{Near optimal constructions of frameproof codes}
\author{Miao Liu\thanks{M. Liu is with the Research Center for Mathematics and Interdisciplinary Sciences, Shandong University, Qingdao 266237, China (e-mail: liumiao10300403@163.com).}, Zengjiao Ma\thanks{Z. Ma is with Key Laboratory of Cryptologic Technology and Information Security of Ministry of Education,
School of Cyber Science and Technology, Shandong University, Qingdao 266237, China (e-mail: mazj0020@foxmail.com).}, and Chong Shangguan\thanks{C. Shangguan is with the Research Center for Mathematics and Interdisciplinary Sciences, Shandong University, Qingdao 266237, China, and the Frontiers Science Center for Nonlinear Expectations, Ministry of Education, Qingdao 266237, China (e-mail: theoreming@163.com).}}
\date{}
\maketitle

\begin{abstract}
\noindent 
Frameproof codes are a class of secure codes that were originally introduced in the pioneering work of Boneh and Shaw in the context of digital fingerprinting. They can be used to enhance the security and credibility of digital content. Let $M_{c,l}(q)$ denote the largest cardinality of a $q$-ary $c$-frameproof code with length $l$. Based on an intriguing observation that relates $M_{c,l}(q)$ to the renowned Erd\H{o}s Matching Conjecture in extremal set theory, in 2003, Blackburn posed an open problem on the precise value of the limit $R_{c,l}=\lim_{q\rightarrow\infty}\frac{M_{c,l}(q)}{q^{\lc l/c \rc}}$. By combining several ideas from the probabilistic method, we present a lower bound for $M_{c,l}(q)$, which, together with an upper bound of Blackburn, completely determines $R_{c,l}$ for {\it all} fixed $c,l$, and resolves the above open problem in the full generality. We also present an improved upper bound for $M_{c,l}(q)$.
\end{abstract}


\section{Introduction}
\noindent For positive integers $q,l$, a $q$-ary code of length $l$ is a subset of $[q]^l$, where $[q]=\{1,\ldots,q\}$. A vector $x\in [q]^l$ is denoted by $x=(x_1,\ldots,x_l)$. For an integer $c\ge 1$ and $c$ vectors $\{x^1,\ldots,x^c\}\subseteq [q]^l$, the set of {\it descendants} $desc(\{x^1,\ldots,x^c\})$ consists of all vectors $y\in[q]^l$ such that for all $i\in[l]$, $y_i\in\{x^1_i,\ldots,x^c_i\}$. A code $\ma{C}\s [q]^l$ is {\it $c$-frameproof} or a {\it $(c,l,q)$-frameproof code} if for every set of $c+1$ distinct codewords $\{x^0,x^1,\ldots,x^c\}\subseteq\ma{C}$, it holds that $x^0\not\in desc(\{x^1,\ldots,x^c\})$; or equivalently, $desc(\{x^1,\ldots,x^c\})\cap\ma{C}=\{x^1,\ldots,x^c\}$.

Frameproof codes were originally introduced in the pioneering work of Boneh and Shaw \cite{boneh1998collusion} in the context of digital fingerprinting, where they were used as a technique to protect copyrighted materials from unauthorized use. In the following, we will illustrate a typical scenario in which frameproof codes can be used. Suppose that a set $\ma{C}\subseteq [q]^l$ of distinct vector labels, known as digital fingerprints, are inserted into $|\mathcal{C}|$ copies of some copyrighted data, before the distributor wants to sell the copies to different users. Loosely speaking, whenever a coalition of $c$ users wants to produce a pirate copy of the digital data, they may compare the coordinates of their fingerprints, say $x^1,\ldots,x^c\in\mathcal{C}$, and produce a pirate copy by taking some vector $x^0\in desc(\{x^1,\ldots,x^c\})$. If $x^0\in\mathcal{C}$, then once the pirate copy is confiscated, the innocent user who holds the fingerprint $x^0$ may be thought to be a member of the coalition; in other words, this user is framed as an illegal user. It is clear by definition that a set $\ma{C}$ of fingerprints is $c$-frameproof if and only if any coalition of at most $c$ users can not frame another user not in this coalition.

Staddon, Stinson, and Wei \cite{staddon2001combinatorial} demonstrated that frameproof codes can also be applied to the traitor tracing schemes studied by Chor, Fiat, and Naor \cite{chor2000tracing}. Frameproof codes are also closely related to a number of well-studied secure codes and combinatorial objects, as briefly mentioned below.
Separable codes are ``weaker'' than frameproof codes in the sense that if a code is $c$-frameproof then it is also $c$-separable. They have found several applications in multimedia fingerprinting; see, e.g., \cite{blackburn2015probabilistic,cheng2011separable,gao2014new}.
Parent-identifying codes and traceability codes are ``stronger'' than frameproof codes in the sense that if a code is $c$-parent-identifying or $c$-traceability then it is also $c$-frameproof. They have been applied to traceability schemes; see, e.g. \cite{blackburn2010traceability,staddon2001combinatorial,gu2019probabilistic,gu2017bounds,
      kabatiansky2019traceability,shangguan2018new}.
Separating hash families are generalizations of frameproof codes; see, e.g., \cite{bazrafshan2011bounds,blackburn2008bound,ge2019some,shangguan2016separating,stinson2008generalized}.
Cover-free families are closely related to frameproof codes. Their relations have been studied in several papers; see, e.g., \cite{dyachkov2017cover,ge2019some,staddon2001combinatorial}.

The main problem in this field is to determine for given parameters $c,l,q$, the maximum possible cardinality of a $(c,l,q)$-frameproof code. This problem has been extensively studied in the last two decades; see \cite{blackburn2003frameproof,boneh1998collusion,chee2012improved,cheng2019improved,cheng2011anti,dyachkov2017cover,ge2019some,randriambololona20132,shangguan2017new,
shchukin2016list,staddon2001combinatorial,stinson2000secure,van2014tight,vorob2017bounds,xing2002asymptotic,yang2016new} for some examples. Let $M_{c,l}(q)$ be the largest cardinality of a $(c,l,q)$-frameproof code. To avoid trivialities, we assume that $c\ge 2$ and $q\ge 2$.

Almost 20 years ago, Blackburn \cite{blackburn2003frameproof} had an interesting observation showing that $M_{c,l}(q)$ is closely related to a beautiful extremal set theory problem raised by Erd\H{o}s, as detailed below.

\begin{definition}\label{def:ext-set}
 For nonnegative integers $l\ge t(\la +1)$, let $m(l,t,\la)$ be the maximum size of a $t$-uniform hypergraph on $l$ vertices that does not contain $\la +1$ pairwise disjoint edges.
\end{definition}

Note that we postpone the definition of a uniform hypergraph to \cref{sec:pre}. Erd\H{o}s has a well-known conjecture on the precise value of $m(l,t,\la)$.

\begin{conjecture}[Erd\H{o}s Matching Conjecture, see e.g. \cite{erdHos1965problem}]\label{con:erd-mat}
 If $l\ge t(\la +1)$, then $$m(l,t,\la)=\max\Big\{\bi{l}{t}-\bi{l-\la}{t},\bi{t(\la+1)-1}{t}\Big\}.$$
\end{conjecture}

The following theorem of Blackburn \cite{blackburn2003frameproof} provides an upper bound on $M_{c,l}(q)$ using $m(l,t,\la)$.

\begin{theorem}[see Corollary 9, \cite{blackburn2003frameproof}]\label{thm:blackburn}
Let $0\le\la\le c-1$ be the unique integer that satisfies $l=c(\lc l/c \rc-1)+\la+1$. Then
\begin{equation}\label{eq:Bla-upp}
  M_{c,l}(q)\le\fr{\bi{l}{\lc l/c \rc}q^{\lc l/c \rc}}{\bi{l}{\lc l/c \rc}-m(l,\lc l/c \rc,\la)}+\bi{l}{\lc l/c \rc-1}q^{\lc l/c \rc-1}.
\end{equation}
\end{theorem}

For fixed positive integers $c,l$, we denote the limit
$$R_{c,l}=\lim_{q\rightarrow\infty}\frac{M_{c,l}(q)}{q^{\lc l/c \rc}}.$$ Interestingly, although the limit $R_{c,l}$ was introduced and used before, as far as we know, the existence of such a limit has never been formally proved for general $c,l$ until this paper (see \cref{cor:limit} below). This limit captures the asymptotic behavior of $M_{c,l}(q)$ as $q$ tends to infinity. By \eqref{eq:Bla-upp} it is clear that if the limit exists then
\begin{equation}\label{eq:Rcl}
R_{c,l}\le\fr{\bi{l}{\lc l/c \rc}}{\bi{l}{\lc l/c \rc}-m(l,\lc l/c \rc,\la)}.
\end{equation}

To the best of our knowledge, in the literature, there are only a few cases where the exact values of $M_{c,l}(q)$ or $R_{c,l}$ are known, as summarized below:

\begin{itemize}
  \item  when $l\le c$, $M_{c,l}(q)=l(q-1)$ \cite{blackburn2003frameproof};
  \item  when $c=2$ and $l$ is even, $R_{2,l}=2$ \cite{blackburn2003frameproof};
  \item  when $l\equiv 1 \pmod{c}$, and $q\ge l$ is a prime power, $M_{c,l}(q)=q^{\lc l/c \rc}$, see \cite{cohen2000efficient} for the lower bound and \cite{van2014tight} for the upper bound;
  \item  when $c+1$ is a prime power and $l=c+2$, $R_{c,c+2}=\fr{c+2}{c}$ \cite{chee2012improved}.
\end{itemize}

The known results on the values of $m(l,\lc l/c \rc,\la)$ indicated that in all the cases listed above, \eqref{eq:Rcl} can be achieved with equality. In fact, the following open question was asked by Blackburn in 2003.

\begin{question}[Blackburn, see Section 8 in \cite{blackburn2003frameproof}]\label{que:Bla}
Is $R_{c,l}=\fr{\bi{l}{\lc l/c \rc}}{\bi{l}{\lc l/c \rc}-m(l,\lc l/c \rc,\la)}$ for all fixed integers $c\ge 2$ and $l\ge 2$, where $0\le\la\le c-1$ satisfies $l=c(\lc l/c \rc-1)+\la+1$?
\end{question}

In this paper, we give a complete and affirmative answer to \cref{que:Bla}, as summarized below.

\subsection{Main result}


\noindent Our first result is a probabilistic construction showing that for all fixed $c,l$ and sufficiently large $q$, the coefficient of $q^{\lc l/c \rc}$ in \eqref{eq:Bla-upp} is always tight. Consequently, \eqref{eq:Rcl} always holds with equality.

\begin{theorem}\label{thm:lower-bd}
  For fixed integers $c\ge 2,~l\ge 2$ and any real $\ep>0$, there exists an integer $q_0=q_0(c,l,\ep)$ such that for any $q>q_0$, $$M_{c,l}(q)\ge (1-\ep)\fr{\bi{l}{\lc l/c \rc}}{\bi{l}{\lc l/c \rc}-m(l,\lc l/c \rc,\la)}q^{\lc l/c \rc},$$ where $0\le\la\le c-1$ is the unique integer that satisfies $l=c(\lc l/c \rc-1)+\la+1$.
\end{theorem}

To prove \cref{thm:lower-bd} we apply several powerful combinatorial ideas developed by previous researchers in the study of cover-free families \cite{erdos1985families,frankl1987colored} and hypergraph packings \cite{frankl1987colored,frankl1985near,rodl1985packing} (see \cref{subsec:packing} for details). Our proof is essentially built on a seminal result of Frankl and R\"{o}dl \cite{rodl1985packing,frankl1985near} and  Pippenger \cite{pippenger1989} on the existence of near optimal hypergraph matchings in hypergraphs with certain pseudorandom property (see \cref{thm:hyper-pac} below). The interested reader is referred to Section 4 in \cite{alon2016probabilistic} for more detailed discussion on this topic.

As our construction giving \cref{thm:lower-bd} is probabilistic, it is interesting to determine whether we can achieve the same lower bound by explicit constructions. We will return to this topic in \cref{sec:con}.

Combining Theorems \ref{thm:blackburn} and \ref{thm:lower-bd}, we have

\begin{corollary}\label{cor:limit}
    For fixed positive integers $c,l$,
    $$R_{c,l}=\lim_{q\rightarrow\infty}\fr{M_{c,l}(q)}{q^{\lc l/c \rc}}=\fr{\bi{l}{\lc l/c \rc}}{\bi{l}{\lc l/c \rc}-m(l,\lc l/c \rc,\la)}.$$
\end{corollary}

When $l$ and $c$ are fixed and $q$ tends to infinity, we may regard $m(l,\lc l/c \rc,\la)$ as a constant that can be computed in finite time. However, determining the exact value of $m(l,t,\la)$ for all positive integers $l,t,\la$ is generally believed to be one of the major open problems in extremal set theory. The trivial case $\lambda=0$ implies that there are no nonempty sets in the family, resulting in $m(l,t,0)=0$. For $\la= 1$, the celebrated Erd\H{o}s-Ko-Rado theorem \cite{erdos1961ko-rado} implies that $m(l,t,1)=\binom{l-1}{t-1}$. Erd\H{o}s and Gallai \cite{erdos1959gallai} resolved the $t=2$ case. Frankl \cite{frankl2017matching} provided a complete solution for the case $t=3$. For general $t$, the following results have been established:

\begin{itemize}
    \item when $l\ge (2\la +1)t -\la$, $m(l,t,\la)=\binom{l}{t}-\binom{l-\la}{t}$ \cite{frankl2013improved}.
    \item for sufficiently large $\la$ and $l\ge \frac{5}{3}\la t -\frac{2}{3}\la$, $m(l,t,\la)=\binom{l}{t}-\binom{l-\la}{t}$ \cite{frankl2022kupa}.
    \item for all $(\la +1)t\le l\le (\la +1)(t+\epsilon)$, where $\epsilon$ depends on $t$, $m(l,t,\la)=\binom{t(\la +1)-1}{t}$ \cite{frankl2017}.
\end{itemize}
\noindent One can substitute the above results on $m(l,t,\la)$ into \cref{cor:limit} to obtain the corresponding values of $R_{c,l}$. We omit the details.

Although Blackburn's upper bound \eqref{eq:Bla-upp} is strong enough to determine $R_{c,l}$ (combined with \cref{thm:lower-bd}), it is still not optimal. The next result shows that for large enough $q$, one can remove the lower order term in \eqref{eq:Bla-upp}.

\begin{proposition}\label{pro:upper-bd}
  For $q\ge\fr{\bi{l}{\lc l/c \rc}-m(l,\lc l/c \rc,\la)}{l-\lc l/c \rc+1}$, we have $M_{c,l}(q)\le\fr{\bi{l}{\lc l/c \rc}}{\bi{l}{\lc l/c \rc}-m(l,\lc l/c \rc,\la)}q^{\lc l/c \rc}$, where $0\le\la\le c-1$ is the unique integer that satisfies satisfies $l=c(\lc l/c \rc-1)+\la+1$.
\end{proposition}

It is easy to see that when $\lambda=0$, we have $m(l,\lc l/c \rc,0)=0$ and $l\equiv 1\pmod{c}$. In this case, \cref{pro:upper-bd} indicates that $M_{c,l}(q)\le q^{\lc l/c \rc}$ for sufficiently large $q$. Note that this gives an alternative and more concise proof for the main result of \cite{van2014tight}.

\subsection{Outline of the paper}
\noindent The rest of this paper is organized as follows. In \cref{sec:pre} we will present the notation and lemmas that are required for our proofs. In \cref{subsec:con} we will introduce the connection between frameproof codes and cover-free multipartite hypergraphs. In \cref{subsec:packing} we will present a result in hypergraph packings, which is the main technical lemma of this paper (see \cref{thm:hyper-pac} below). The proofs of this lemma will be postponed to \cref{sec:lem}.
The main results of this paper, \cref{thm:lower-bd} and \cref{pro:upper-bd}, will be proved in \cref{sec:low,sec:upp}, respectively. Finally, we will conclude this paper in \cref{sec:con}.

\section{Preliminaries}\label{sec:pre}

\subsection{Codes and hypergraphs}\label{subsec:con}

\noindent For a positive integer $l$, let $[l]:=\{1,\ldots,l\}$ denote the set of the first $l$ positive integers. For a finite set $X$ with $|X|\ge t$, let $\bi{X}{t}$ denote the family of all $t$-subsets of $X$, namely $\bi{X}{t}:=\{S\s X: |S|=t\}$. For a vector $x\in [q]^l$ and a subset $S\s [l]$, let $x_S:=(x_i:i\in S)\in [q]^{|S|}$ denote the subsequence of $x$ with coordinates indexed by $S$.

For $x\in \ma{C}$ and $S\in\bi{[l]}{t}$, $x_S$ is called an {\it own $t$-subsequence} of $x$ with respect to $\ma{C}$ if for every $y\in\ma{C}\setminus\{x\}$, $y_S\neq x_S$; otherwise, if there is $y\in\ma{C}\setminus\{x\}$ with $y_S=x_S$, then $x_S$ is called a {\it non-own $t$-subsequence} of $x$. A code $\ma{C}$ has no own $t$-subsequence if and only if for every $x\in \ma{C}$ and $S\in\bi{[l]}{t}$, $x_S$ is a non-own $t$-subsequence of $x$.

A {\it hypergraph} $\ma{H}$ is an ordered pair $\ma{H}=(V(\ma{H}),E(\ma{H}))$, where the {\it vertex set} $V(\ma{H})$ is a finite set, and the {\it edge set} $E(\ma{H})$ is a family of distinct subsets of $V(\ma{H})$. A hypergraph is called {\it $l$-uniform} if for every $A\in E(\ma{H})$,  $|A|=l$. A hypergraph is called {\it $l$-partite} if its vertex set $V(\ma{H})$ admits a partition $V(\ma{H})=\cup_{i=1}^l V_i$ such that every edge has at most one intersection with each $V_i$; i.e., for every $A\in E(\ma{H})$ and $i\in[l]$, we have $|A\cap V_i|\le 1$. Note that any hypergraph defined on $l$ vertices is automatically $l$-partite.

\begin{definition}[Complete multi-partite hypergraphs]\label{def:com-mul-hyp}
For all positive integers $q,l$, let $\ma{H}_l(q)$ denote the complete $l$-partite $l$-uniform hypergraph with equal part size $q$, where we assume without loss of generality that $V(\ma{H}_l(q))$ admits a $l$-partition   $V(\ma{H}_l(q))=\cup_{i=1}^l V_i$ and for each $i\in[l]$, $V_i=\{(i,a):a\in [q]\}$. Then
$$E(\ma{H}_l(q))=\big\{\{(1,a_1),\ldots,(l,a_l)\}: a_1,\ldots,a_l\in[q]\big\}.$$ Clearly, $|V(\ma{H}_l(q))|=lq$ and $|E(\ma{H}_l(q))|=q^l$.
\end{definition}

The next lemma connects codes to multi-partite hypergraphs.

\begin{lemma}\label{lem:1-2-1}
  For all positive integers $q,l$, there is a one-to-one correspondence between the family of $q$-ary codes of length $l$ and the family of subhypergraphs of $\ma{H}_l(q)$.
\end{lemma}

\begin{proof}
Let $\pi:[q]^l\rightarrow E(\ma{H}_l(q))$ be a mapping defined as for each vector $x=(x_1,\ldots,x_l)\in [q]^l$,
$$\pi(x):=\{(1,x_1),\ldots,(l,x_l)\}\in E(\ma{H}_l(q)).$$
According to $\pi$, every code $\ma{C}\s [q]^l$ is mapped to a subhypergraph $\pi(\ma{C})$ of $\ma{H}_l(q)$, with $V(\pi(\ma{C}))=V(\ma{H}_l(q))$ and $E(\pi(\ma{C}))=\{\pi(x):x\in\ma{C}\}.$
It is not hard to check that $\pi$ is a bijection, which gives us the promised one-to-one correspondence.
\end{proof}

Similarly to the definition of subsequences, one can define the own $t$-subsets and the non-own $t$-subsets of a $l$-uniform hypergraph. Specifically, for a $l$-uniform hypergraph $\ma{H}$, an edge $A\in E(\ma{H})$, and a subset $T\in\bi{A}{t}$, $T$ is called an {\it own $t$-subset} of $A$ if for every $B\in E(\ma{H})\setminus\{A\}$, $T\nsubseteq B$; otherwise, $T$ is called a {\it non-own $t$-subset} of $A$. A hypergraph $\ma{H}$ has no own $t$-subset if and only if for every $A\in E(\ma{H})$ and $T\in\bi{A}{t}$, $T$ is a non-own $t$-subset of $A$.

It is not hard to verify the following lemma.

\begin{lemma}\label{fact-1}
A code $\ma{C}$ has no own $t$-subsequence if and only if the hypergraph $\pi(\ma{C})$ has no own $t$-subset.
\end{lemma}

\begin{proof}
  With the notation above, for $x\in\ma{C}$ and $S\in\bi{[l]}{t}$ one can define the restriction of $\pi$ to $S$ by $\pi_S(x):=\{(i,x_i):i\in S\}$. Therefore, it is not hard to observe that $x_S$ is an own $t$-subsequence of $x$ with respect to $\ma{C}$ if and only if $\pi_S(x):=\{(i,x_i):i\in S\}$ is an own $t$-subset of $\pi(x)$ with respect to $\pi(\ma{C})$. The lemma is an immediate consequence of this observation.
\end{proof}

A $l$-uniform hypergraph $\ma{H}$ is called {\it $c$-cover-free} by Erd\H{o}s, Frankl and F\"uredi \cite{erdos1985families} if for every set of $c+1$ distinct edges $\{A_0,A_1,\ldots,A_c\}\subseteq E(\ma{H})$, it holds that $A_0\nsubseteq \cup_{i=1}^c A_i$.

The following lemma connects frameproof codes to cover-free multi-partite hypergraphs.

\begin{lemma}\label{lem:connection}
  A code $\ma{C}$ is $c$-frameproof if and only if the hypergraph $\pi(\ma{C})$ is $c$-cover-free.
\end{lemma}

\begin{proof}
  A code $\ma{C}$ is $c$-frameproof if and only if for all distinct codewords $x^0,x^1,\ldots,x^{c+1}\in\ma{C}$, $x^0\not\in desc(x^1,\ldots,x^{c})$, and this happens if and only if $\pi(x^0)\nsubseteq \cup_{i=1}^c \pi(x^i)$.
\end{proof}

Note that in the remaining part of this paper, we will use $x\in[q]^l$ and $\pi(x)\in E(\ma{H}_l(q))$ (resp. $\mathcal{C}\subseteq [q]^l$ and $\pi(\mathcal{C})\subseteq E(\ma{H}_l(q))$) interchangeably.

\subsection{Hypergraph packings}\label{subsec:packing}

\noindent For a fixed $l$-uniform hypergraph $\ma{F}$ and a {\it host} $l$-uniform hypergraph $\ma{H}$, an {\it $\ma{F}$-packing} in $\ma{H}$ is a family of edge disjoint copies of $\ma{F}$ in $\ma{H}$, as formally defined below.

\begin{definition}[Hypergraph packing]\label{def:hyp-pac}
  A family of $m$ $l$-uniform hypergraphs
  \begin{equation}\label{eq:hyp-pac}
    \{\ma{F}_1=(V(\ma{F}_1),E(\ma{F}_1)),\ldots,\ma{F}_m=(V(\ma{F}_m),E(\ma{F}_m))\}
  \end{equation}
forms a $\ma{F}$-packing in $\ma{H}$ if for each $j\in[m]$,
\begin{enumerate}
  \item $V(\ma{F}_j)\s V(\ma{H})$, $E(\ma{F}_j)\s E(\ma{H})$;
  \item $\ma{F}_j$ is a copy of $\ma{F}$;
  \item $E(\ma{F}_j)\cap E(\ma{F}_{j'})=\emptyset$ for distinct $j,j'\in [m]$.
\end{enumerate}
\end{definition}

Suppose $\ma{F}$ and $\ma{H}$ are both $l$-partite, say, $V(\ma{F})$ and $V(\ma{H})$ have $l$-partitions $V(\ma{F})=\cup_{i=1}^l W_i$ and $V(\ma{H})=\cup_{i=1}^l V_i$, respectively. In this case, a copy $\ma{F}'$ of $\ma{F}$ in $\ma{H}$ is called {\it faithful} (with respect to the partitions above) if for each $i\in[l]$, the copy of $W_i$ in $V(\ma{F}')$ is contained in $V_i$. For $l$-partite hypergraphs $\ma{F}$ and $\ma{H}$ with given $l$-partitions of $V(\mathcal{F})$ and $V(\mathcal{H})$, the $\ma{F}$-packing in \eqref{eq:hyp-pac} is said to be {\it faithful} if it further satisfies for every $j\in [m]$, $\ma{F}_j$ is a faithful copy of $\ma{F}$.


For a $l$-uniform hypergraph $\ma{H}$ and $t\in[l]$, let $\ma{H}^{(t)}$ be the $t$-uniform hypergraph defined as $V(\ma{H}^{(t)})=V(\ma{H})$ and
\begin{equation}\label{eq:H^t}
  E(\ma{H}^{(t)})=\bigcup_{A\in E(\ma{H})}\bi{A}{t}.
\end{equation}
As an example, $\ma{H}_l^{(t)}(q)$ is the complete $l$-partite $t$-uniform hypergraph with $\bi{l}{t}q^t$ edges. With the notation of \cref{def:com-mul-hyp}, $$E(\ma{H}_l^{(t)}(q))=\big\{\{(i_1,a_{i_1}),\ldots,(i_t,a_{i_t})\}: 1\le i_1<\ldots <i_t\le l,~ a_{i_1},\ldots,a_{i_t}\in[q]\big\}.$$

To prove \cref{thm:lower-bd}, we will need a slightly more sophisticated notion of hypergraph packing, which we call an {\it induced hypergraph packing}. Note that this notion was originally introduced in \cite{frankl1987colored}.

\begin{definition}[Induced hypergraph packing]\label{def:indu-packing}
    Let $\mathcal{F}$ and $\mathcal{H}$ be $t$-uniform hypergraphs. An $\mathcal{F}$-packing
    $$\{\ma{F}_1=(V(\ma{F}_1),E(\ma{F}_1)),\ldots,\ma{F}_m=(V(\ma{F}_m),E(\ma{F}_m))\}$$ in $\mathcal{H}$ is called an induced $\mathcal{F}$-packing if it satisfies
    \begin{enumerate}
        \item $|V(\mathcal{F}_i)\cap V(\mathcal{F}_j)|\le t$ for $i\neq j$.
        \item For $i\neq j$, let $A=V(\mathcal{F}_i)\cap V(\mathcal{F}_j)$. If $|A|=t$, then $A\notin E(\mathcal{F}_i)$ and $A\notin E(\mathcal{F}_j)$.
    \end{enumerate}
\end{definition}
\noindent Similarly, a induced $\mathcal{F}$-packing is called a {\it faithful induced $\mathcal{F}$-packing} if it is also faithful.

It is clear by definition that every family of $\mathcal{F}$-packing in $\mathcal{H}_l^{(t)}(q)$ contains at most $\frac{\bi{l}{t}q^t}{|E(\ma{F})|}$ edge disjoint copies of $\mathcal{F}$. The next lemma, which is the main technical lemma in this paper, shows that $\mathcal{H}_l^{(t)}(q)$ has a faithful induced $\mathcal{F}$-packing with near optimal size.

\begin{lemma}\label{thm:hyper-pac}
  For a real $\ep>0$ and a $t$-uniform hypergraph $\ma{F}=([l],E(\ma{F}))$, there exists an integer $q_0=q_0(t,l,\ep)$ such that for any $q>q_0$, $\ma{H}_l^{(t)}(q)$ has a faithful induced $\ma{F}$-packing with cardinality at least $(1-\ep)\fr{\bi{l}{t}q^t}{|E(\ma{F})|}$.
\end{lemma}


The proof of \cref{thm:hyper-pac} will be postponed to \cref{sec:lem}.

\section{Proof of \cref{thm:lower-bd}}\label{sec:low}

\noindent Let $\ma{A}$ be a largest $\lc l/c \rc$-uniform hypergraph on the vertex set $[l]$ that does not contain $\la+1$ pairwise disjoint edges. Then by \cref{def:ext-set}, $|E(\ma{A})|=m(l,\lc l/c \rc,\la)$. Let $\ma{F}=([l],E(\ma{F}))$ be the $\lc l/c \rc$-uniform hypergraph with $E(\ma{F})=\bi{[l]}{\lc l/c \rc}\setminus E(\ma{A})$. It is clear that $|E(\ma{F})|=\bi{l}{\lc l/c \rc}-m(l,\lc l/c \rc,\la)$. Note that we will view $\ma{F}$ as a $l$-partite $\lc l/c \rc$-uniform hypergraph.

According to \cref{thm:hyper-pac}, for $q>q_0(\lc l/c \rc,l,\ep)$, $\ma{H}_l^{\lc l/c \rc}(q)$ has a faithful induced $\ma{F}$-packing $$\ma{F}_1=(U_1,E(\ma{F}_1)),\ldots,\ma{F}_m=(U_m,E(\ma{F}_m)),$$
where $|U_i|=l$ for $1\le i\le m$ and $m\ge(1-\ep)\fr{\bi{l}{\lc l/c \rc}q^{\lc l/c \rc}}{|E(\ma{F})|}$. Since the above $\ma{F}$-packing is faithful, one can assume that $U_i\s V(\ma{H}_l(q))=\cup_{j=1}^l V_j$ and $|U_i\cap V_j|=1$ for all $i\in[m],j\in[l]$, where $V_j=\{(j,a):a\in[q]\}$.

Let $\ma{C}=(V(\ma{H}_l(q)),E(\ma{C}))$ be the $l$-partite $l$-uniform hypergraph with $E(\ma{C})=\{U_1,\ldots,U_m\}$. By \cref{lem:1-2-1}, $\ma{C}$ can be equivalently viewed as a $q$-ary code of length $l$. To conclude the proof, it is enough to show that $\ma{C}$ is $c$-frameproof. By \cref{lem:connection} it suffices to show that $\ma{C}$ is $c$-cover-free.
Suppose that this is false and assume without loss of generality that $\{U_1,\ldots,U_{c+1}\}$ violates the $c$-frameproof (or equivalently, $c$-cover-free) property, say, $U_{c+1}\s\cup_{i=1}^c U_i$.
Since the $\mathcal{F}$-packing is induced, by \cref{def:indu-packing} 1) we have that for every $1\le i\le c$, $|U_{c+1}\cap U_i|\le \lc l/c \rc$.
Since $|U_{c+1}|=l=c(\lc l/c \rc-1)+\la+1$, it is not hard to check that there must exist distinct $U_{j_1},\ldots,U_{j_{\la+1}}\in\{U_1,\ldots,U_c\}$ such that $|U_{c+1}\cap U_{j_i}|=\lc l/c \rc$ for every $1\le i\le \la+1$. Moreover, the $\la+1$ $\lc l/c \rc$-subsets $U_{c+1}\cap U_{j_i}$ are pairwise disjoint. It follows by \cref{def:indu-packing} 2) that $U_{c+1}\cap U_{j_i}\not\in E(\ma{F}_{c+1})$ for each $1\le i\le \la+1$. Consequently, the $\lc l/c \rc$-uniform hypergraph $\ma{A}_{c+1}=(U_{c+1},\binom{U_{c+1}}{\lc l/c \rc}\setminus E(\ma{F}_{c+1}))$, which is clearly a copy of $\ma{A}$, contains $\la+1$ pairwise disjoint edges $U_{c+1}\cap U_{j_i}$, which is an obvious contradiction to the definition of $\ma{A}$. Thus, the proof of the theorem is complete.

\section{Proof of \cref{pro:upper-bd}}\label{sec:upp}
\noindent Let us first prove the following lemma.

\begin{lemma}\label{lem:upper-of-own-subsets}
  Let $\lambda$ be defined in \cref{pro:upper-bd}. If a $(c,l,q)$-frameproof code has no own $(\lc l/c \rc-1)$-subsequence, then every codeword of it must contain at least $\bi{l}{\lc l/c \rc}-m(l,\lc l/c \rc,\la)$ own $\lc l/c \rc$-subsequences.
\end{lemma}

\begin{proof}
Let $\ma{C}_1$ be a code that satisfies the assumption of the lemma. We proceed to show that any codeword of $\ma{C}_1$ contains at most $m(l,\lc l/c \rc,\la)$ non-own $\lc l/c \rc$-subsequences. Assume by contradiction that there is a codeword $x\in\ma{C}_1$ that contains at least $m:=m(l,\lc l/c \rc,\la)+1$ non-own $\lc l/c \rc$-subsequences, say $x_{T_1},\ldots,x_{T_m}$, where for each $i\in[m]$, $T_i\in\bi{[l]}{\lc l/c \rc}$. As $m$ is strictly larger than $m(l,\lc l/c \rc,\la)$, by \cref{def:ext-set}, the set system $\{T_1,\ldots,T_m\}\subseteq\bi{[l]}{\lc l/c \rc}$ must contain $\la+1$ pairwise disjoint members. Assume without loss of generality that the first $\la+1$ members $T_1,\ldots,T_{\la+1}$ are pairwise disjoint. By the definition of a non-own subsequence, there exist (not necessarily distinct) codewords $y^1,\ldots,y^{\la+1}\in \ma{C}_1\setminus \{x\}$ such that for each $1\le i\le \la+1$, $y^i_{T_i}=x_{T_i}$.

As $l=(\la+1)\lc l/c \rc+(c-\la-1)(\lc l/c \rc-1)$, there exist $c-\la-1$ (possibly $c-\la-1=0$) pairwise disjoint subsets $S_1,\ldots,S_{c-\la-1}\subseteq\bi{[l]}{\lc l/c \rc-1}$ such that the $c$ subsets $$\{T_i,S_j:i\in [\la+1], j\in [c-\la-1]\}$$

\no form a partition of $[l]$. Since by assumption $\ma{C}_1$ has no own $(\lc l/c \rc-1)$-subsequence, there exist (not necessarily distinct) $z^1,\ldots,z^{c-\la-1}\in \ma{C}_1\setminus \{x\}$ such that for each $j\in [c-\la-1]$, $z^j_{S_j}=x_{S_j}$. As a result, $$x\in desc(y^1,\ldots,y^{\la+1},z^1\ldots,z^{c-\la-1}),$$ violating the assumption that $\ma{C}_1$ is $c$-frameproof. This leads to a contradiction. 
\end{proof}

We are now in a position to prove \cref{pro:upper-bd}.

\begin{proof}[Proof of \cref{pro:upper-bd}]
  For a $(c,l,q)$-frameproof code $\ma{C}$, let $\ma{C}_0$ be the set of codewords that have at least one own subsequence with size $\lc l/c \rc-1$, and let $\ma{C}_1=\ma{C}\setminus\ma{C}_0$. By \cref{fact-1}, for each $x\in\ma{C}_0$, $\pi(x)$ has at least one own $(\lc l/c \rc-1)$-subset. For each $x\in\ma{C}_0$, let
  $$\ma{N}_x:=\{T\in E(\ma{H}_l^{(\lc l/c \rc-1)}(q)): \text{$T$ is an own $(\lc l/c \rc-1)$-subset of $\pi(x)$ with respect to $\pi(\mathcal{C})$}\}. $$
  Then $\ma{N}_x\neq\emptyset$ for every $x\in\ma{C}_0$. Let
  $$\ma{A}_x:=\{S\in E(\ma{H}_l^{(\lc l/c \rc)}(q)): \exists~ T\in \ma{N}_x \text{ such that } T\subseteq S\}.$$
  It is easy to see that $\ma{A}_x$ is the family of $\lc l/c \rc$-subsets each of which contains at least one own $(\lc l/c \rc-1)$-subset of $\pi(x)$. As every $T\in E(\ma{H}_l^{(\lc l/c \rc-1)}(q))$ is contained in exactly $(l-\lc l/c \rc+1)q$ edges in $E(\ma{H}_l^{(\lc l/c \rc)}(q))$ and $\ma{N}_x\neq\emptyset$, for every $x\in\ma{C}_0$,
  \begin{equation}\label{eq:A_x}
     |\ma{A}_x|\ge(l-\lc l/c \rc+1)q.
  \end{equation}

  \noindent Next, for each $y\in\ma{C}_1$, let $\ma{B}_y\s \ma{H}_l^{(\lc l/c \rc)}(q)$ be the family of own $\lc l/c \rc$-subsets of $y$. Then, by \cref{lem:upper-of-own-subsets} we have
  \begin{equation}\label{eq:B_x}
     |\ma{B}_y|\ge\bi{l}{\lc l/c \rc}-m(l,\lc l/c \rc,\la).
  \end{equation}

  It is not hard to check by the definition of own subsets that the $|\ma{C}|$ families $\ma{A}_x,\ma{B}_y, x\in\ma{C}_0,y\in\ma{C}_1$ are pairwise disjoint, as subfamilies of $E(\ma{H}_l^{(\lc l/c \rc)}(q))$. Consequently, by \eqref{eq:A_x} and \eqref{eq:B_x} we have that for large enough $q$,
  \begin{equation*}
    \begin{aligned}
      |E(\ma{H}_l^{(\lc l/c \rc)}(q))|&=\bi{l}{\lc l/c \rc}q^{\lc l/c \rc}
      \ge\sum_{x\in\ma{C}_0}|\ma{A}_x|+\sum_{y\in\ma{C}_1}|\ma{B}_x|\\
      &\ge (l-\lc l/c \rc+1) q |\ma{C}_0|+\Big(\bi{l}{\lc l/c \rc}-m(l,\lc l/c \rc,\la)\Big)|\ma{C}_1|\\
      &\ge \Big(\bi{l}{\lc l/c \rc}-m(l,\lc l/c \rc,\la)\Big)|\ma{C}|,
    \end{aligned}
  \end{equation*}
  and therefore the proposition holds immediately.
\end{proof}

\section{Proof of \cref{thm:hyper-pac}}\label{sec:lem}

\noindent The main concern of this section is to present the proof of \cref{thm:hyper-pac}. Let us begin with some preparations. For a hypergraph $\ma{H}=(V(\ma{H}),E(\ma{H}))$, a subset $\ma{M}\s E(\ma{H}
)$ is a {\it matching} if for every pair of distinct edges $A,B\in\ma{M}$, we have $A\cap B=\emptyset$.
For $x,y\in V(\ma{H})$, the {\it degree} of $x$, denoted by $d_\ma{H}(x)$, is the number of edges that contain $x$; i.e., $d_\ma{H}(x)=|\{A\in E(\ma{H}): x\in A\}|$. The {\it codegree} of $x,y$, denoted by $d_\ma{H}(x,y)$, is the number of edges that contain both $x$ and $y$; i.e., $d_\ma{H}(x,y)=|\{A\in E(\ma{H}): \{x,y\}\s A\}|$. Let $\Delta(\ma{H}),\delta(\ma{H}),\Delta_2(\ma{H})$ denote the maximum degree, minimum degree, and maximum codegree of $\ma{H}$, respectively. 

We will use a well-known result developed in a series of works of R\"{o}dl \cite{rodl1985packing}, Frankl and R\"{o}dl \cite{frankl1985near}, and Pippenger \cite{pippenger1989}, which states a sufficient condition for which a hypergraph admits a near optimal matching.

\begin{theorem}[see, e.g., Theorem 4.7.1 in \cite{alon2016probabilistic}]\label{lem:pippenger}
    For every integer $t\ge 2$ and reals $k\ge 1$ and $\ep>0$, there are $\gamma=\gamma(t,k,\ep)>0$ and $d_0=d_0(t,k,\ep)$ such that for every $n$ and $D\ge d_0$ the following holds. Every $t$-uniform hypergraph $\mathcal{H}=(V,E(\mathcal{H}))$ with $|V|=n$ which satisfies the following conditions
    \begin{itemize}
        \item [(i)] \label{pip(1)} for all but at most $\gamma n$ vertices $x\in V$, $|d_\ma{H}(x)-D|\le \gamma D$;
        \item [(ii)]  $\Delta(\ma{H})<kD$;
        \item [(iii)]  $\Delta_2(\ma{H})<\gamma D$;
    \end{itemize}
     contains a matching that covers all but at most $\ep n$ vertices.
\end{theorem}

The following lemma is a direct consequence of \cref{lem:pippenger}.

\begin{lemma}\label{lem:pippenger-app}
Given integers $s$, $l$ and a real $\ep>0$, for any $t\in[l]$,  there exists an integer $q_0$
so that for every integer $q>q_0$ there exists a faithful $\ma{H}_l^{(t)}(s)$-packing
$$\ma{F}_1=(U_1,E(\ma{F}_1)),\ldots,\ma{F}_m=(U_m,E(\ma{F}_m))$$
\no in $\ma{H}_l^{(t)}(q)$ with $m\ge(1-\ep) (\fr{q}{s})^t$, where for each $j\in[m]$, $\ma{F}_j$ is a faithful copy of $\ma{H}_l^{(t)}(s)$ in $\ma{H}_l^{(t)}(q)$.
\end{lemma}

\begin{proof}
    By \eqref{eq:H^t} it is clear that $\ma{H}_l^{(t)}(s)$ is a $l$-partite $t$-uniform hypergraph with $\bi{l}{t}s^t$ edges and $\ma{H}_l^{(t)}(q)$ is a $l$-partite $t$-uniform hypergraph with $\bi{l}{t}q^t$ edges. It is routine to check that, if $\ma{F}$ and $\ma{H}$ are both $l$-partite, and moreover $\ma{F}$ is complete, then every copy of $\ma{F}$ in $\ma{H}$ must be faithful. Therefore, every $\ma{H}_l^{(t)}(s)$-packing in $\ma{H}_l^{(t)}(q)$ is faithful.
    Let $\ma{H}$ be the $\bi{l}{t}s^t$-uniform hypergraph formed by all copies of $\ma{H}_l^{(t)}(s)$ in $\ma{H}_l^{(t)}(q)$. That is, the vertex set $V(\ma{H})$ is the set of the $\bi{l}{t}q^t$ edges of $\ma{H}_l^{(t)}(q)$, and $\bi{l}{t}s^t$ vertices in $V(\ma{H})$ form an edge if and only if they form a (faithful) copy of $\ma{H}_l^{(t)}(s)$.

To apply \cref{lem:pippenger}, we will show that $\mathcal{H}$ fulfills the assumptions of the theorem. Let $D:=\bi{q-1}{s-1}^t \binom{q}{s}^{l-t}$. We claim that for any $x\in V(\ma{H})=E(\ma{H}_l^{(t)}(q))$,
\begin{align*}
    d_\ma{H}(x)=|\{\text{number of copies of $\ma{H}_l^{(t)}(s)$ in $\ma{H}_l^{(t)}(q)$ containing $x$}\}|=D.
\end{align*}
\noindent Indeed, assume that $V(\ma{H}_l^{(t)}(q))=\cup_{i=1}^l V_i$, where the $V_i$'s are given in \cref{def:com-mul-hyp}; given $x\in V(\mathcal{H})$, say $x\in\cup_{i=1}^t V_i$, to form a copy of $\ma{H}_l^{(t)}(s)$, one needs to choose $s-1$ vertices from each $V_i,1\le i\le t$, and $s$ vertices from each $V_j, t+1\le j\le l$, which gives us $D$ choices in total.

We now show that $\Delta_2(\mathcal{H})$ is small. Observe that for any pair of distinct $x,y\in V(\ma{H})$, we have $|x\cup y|\ge t+1$. Similarly to the computation of $d_\ma{H}(x)$, we have that
\begin{align*}
    d_\ma{H}(x,y)\le\bi{q-1}{s-1}^{|x\cup y|} \binom{q}{s}^{l-|x\cup y|}\le\bi{q-1}{s-1}^{t+1} \binom{q}{s}^{l-t-1}=\frac{s}{q}\cdot D.
\end{align*}

\noindent For $\gamma$ that meets the assumption of \cref{lem:pippenger} and sufficiently large $q$, we have $\Delta_2(\mathcal{H})\le\gamma D$. Therefore, it follows by \cref{lem:pippenger} that for sufficiently $q$, $\ma{H}$ has a matching with cardinality at least $(1-\ep)(\fr{q}{s})^t$. As a matching in $\ma{H}$ corresponds to a (faithful) $\ma{H}_l^{(t)}(s)$-packing in $\ma{H}_l^{(t)}(q)$, the lemma follows immediately.
\end{proof}


We will also need the following standard probabilistic inequalities. Note that all random variables considered below are discrete.


\begin{proposition}[see e.g. \cite{alon2016probabilistic}]\label{prop:ineq}~
\begin{itemize}
\item [(i)] Markov inequality: If $X$ is a nonnegative random variable and $a>0$, then
$$
\Pr[X \geq a] \leq \frac{\E[X]}{a} .
$$


\item [(ii)] Hoeffding inequality: Let $X_1, \ldots, X_n$ be independent random variables taking values in $[a,b]$. Consider the sum $X=\sum_{i=1}^n X_i$ and denote the expectation of the sum by $\mu=\E[X]$. Then for any $t>0$,
$$\Pr[|X-\mu|\ge t]\le 2e^{-\frac{2t^2}{n(b-a)^2}}. $$


%
\end{itemize}
\end{proposition}

Now we are ready to present the proof of \cref{thm:hyper-pac}.

\subsection{Proof of \cref{thm:hyper-pac}}

\noindent Recall that we assume $V(\ma{H}_l^{(t)}(q))=\cup_{i=1}^l V_i$, where the $V_i$'s are given in \cref{def:com-mul-hyp}. By \cref{lem:pippenger-app}, for every $\delta>0$, there exists $q_0$ such that for every $q>q_0$, there exists a faithful $\ma{H}_l^{(t+1)}(1)$-packing $\ma{P}$
in $\ma{H}_l^{(t+1)}(q)$ with $|\ma{P}|\ge(1-\delta) q^{t+1}$. Note that every member in $\mathcal{P}$ is of the form $(U,\binom{U}{t+1})$, where $|U|=l$ and $|U\cap V_i|=1$ for each $i\in[l]$. With slight abuse of notation, we will write $U\in\mathcal{P}$ to denote the copy of $\ma{H}_l^{(t)}(q)$ in $\mathcal{P}$ defined on the vertex set $U$. Then, for distinct $U,U'\in\mathcal{P}$, $|U\cap U'|\le t$.

For some fixed but sufficiently small $\eta>0$, let $\mathcal{R}\subseteq E(\ma{H}_l^{(t)}(q))$ be a random set formed by choosing every element of $E(\ma{H}_l^{(t)}(q))$ independently with probability $1-\eta$. Then $|\ma{R}|$ can be viewed as the sum of a series of independent Bernoulli random variables. As $\E[|\mathcal{R}|]=(1-\eta)|E(\ma{H}_l^{(t)}(q))|=(1-\eta)\binom{l}{t}q^t$, it follows by \cref{prop:ineq} (ii) that
for large enough $q$, with high probability (by which we mean that with probability $1-o(1)$, where $o(1)\rightarrow 0$ as $q\rightarrow\infty$) we have
\begin{equation}\label{Rcardinality}
\left||\mathcal{R}|-(1-\eta)\binom{l}{t}q^t\right|\le\eta\binom{l}{t}q^t.
\end{equation}

\noindent Let $\mathcal{H}$ be an $|E(\ma{F})|$-uniform hypergraph with vertex set $\mathcal{R}$, where $\mathcal{A}:=\{A_1,\ldots,A_{|E(\ma{F})|}\}\subseteq\mathcal{R}$ forms an edge in $\mathcal{H}$ if $\mathcal{A}$ forms a copy of $\mathcal{F}$ and there exists $U\in\mathcal{P}$ such that $\mathcal{A}=\binom{U}{t}\cap\mathcal{R}$.

The proof of \cref{thm:hyper-pac} is divided into two parts. First, the following lemma shows that every matching in $\mathcal{H}$ defines a faithful induced $\mathcal{F}$-packing in $\mathcal{H}_l^{(t)}(q)$. Then, \cref{lem:optimal-matching} below shows that $\mathcal{H}$ has a sufficiently large matching.

\begin{lemma}\label{lem:packing-matching}
Every matching in $\mathcal{H}$ defines a faithful induced $\mathcal{F}$-packing in $\mathcal{H}_l^{(t)}(q)$.
\end{lemma}

\begin{proof}
Let $\ma{M}=\{\mathcal{A}_1,\ldots,\mathcal{A}_m\}$ be a matching in $\mathcal{H}$. By the definition of $\mathcal{H}$, there exist $U_1,\dots,U_m\in\mathcal{P}$ such that for each $1\le i\le m$, $\mathcal{A}_i=\binom{U_i}{t}\cap \mathcal{R}$. Consider the family of $\mathcal{F}$-copies in $\mathcal{H}_l^{(t)}(q)$,
$$\{\mathcal{F}_i=(U_i
,\mathcal{A}_i):1\le i\le m\}.$$

\noindent It suffices to show that $\mathcal{F}_1,\ldots,\mathcal{F}_m$ form a faithful induced $\mathcal{F}$-packing.

First, since $\mathcal{P}$ is a faithful $\mathcal{H}_l^{(t+1)}(1)$-packing, each $\mathcal{F}_i=(U_i
,\mathcal{A}_i)$ is also faithful and for all $i\ne j$, $|U_i\cap U_j|\le t$. Second, suppose that there exist $i\neq j$ such that $A=U_i\cap U_j$ and $|A|= t$. We will show that $A\not\in\mathcal{A}_i$ and $A\not\in\mathcal{A}_j$. Suppose for contradiction that $A\in\mathcal{A}_i=\binom{U_i}{t}\cap \mathcal{R}$. Then, by the definition of $A$, it also belongs to $\binom{U_j}{t}\cap \mathcal{R}=\mathcal{A}_j$. It follows that $\mathcal{A}_i\cap\mathcal{A}_j\neq\emptyset$, which is a contradiction since $\mathcal{A}_i$ and $\mathcal{A}_j$ are contained in a matching in $\mathcal{H}$. Therefore, by \cref{def:indu-packing}, the $\mathcal{F}_i$'s form a faithful induced $\mathcal{F}$-packing, as needed.
\end{proof}

To complete the proof of \cref{thm:hyper-pac}, it suffices to show that $\mathcal{H}$ has a near optimal matching.

\begin{lemma}\label{lem:optimal-matching}
$\mathcal{H}$ has a matching with size at least $(1-o(1))\frac{\binom{l}{t}q^t}{|E(\mathcal{F})|}$,	
where $o(1)\rightarrow 0$ as $q\rightarrow \infty.$
\end{lemma}

\begin{proof}
We will prove the lemma by applying \cref{lem:pippenger}. For that purpose, we will show that $\mathcal{H}$ satisfies assumptions (i)-(iii) of \cref{lem:pippenger}.

First of all, we will show that $\mathcal{H}$ is almost regular, that is, it satisfies \cref{lem:pippenger} (i).
 For every $A\in V(\mathcal{H})=\mathcal{R}$, it follows by the definition of $\mathcal{H}$ that
    $$d_{\mathcal{H}}(A)=|\{U\in\mathcal{P}:A\subseteq U \text{ and }\binom{U}{t}\cap\mathcal{R}\text{ is a copy of }\mathcal{F}\}|.$$
	For each $U\in \ma{P}$ with $A \subseteq U$, let $X_{A,U}$ denote the indicator random variable for the event that $\binom{U}{t}\cap\mathcal{R}$ forms a copy of $\ma{F}$. Note that $\Pr[X_{A,U}=1] = \lambda_{\mathcal{F}}(1-\eta)^{|E(\ma{F})|-1}\eta^{\binom{l}{t}-|E(\ma{F})|}$, where $\lambda_{\mathcal{F}}$ is the number of nonequivalent embeddings of $\ma{F}$ into $\bi{U}{t}$ with one edge fixed. For convenience, let $p$ denote the above probability. By definition, $$d_\ma{H}(A)=\sum_{U\in \ma{P}:A \subseteq U}X_{A,U}.$$
    Observe that by the property of $\mathcal{P}$, for a fixed $A$ the family $\{X_{A,U}: A\subseteq U\in\mathcal{P}\}$ of random variables are mutually independent.

     Let $d_{\ma{P}}(A)=|\{U\in\mathcal{P}:A\subseteq U \}|$. Then $$\mathbb{E}[d_\ma{H}(A)]=\sum_{U\in \ma{P}:A \subseteq U}\Pr[X_{A,U}=1]=d_{\ma{P}}(A)\cdot p.$$

    We will make use of the following claim, whose proof is postponed to the end of this section.

     \begin{claim}\label{claim}
     For every $A\in E(\ma{H}_l^{(t)}(q))$, $|d_{\mathcal{P}}(A)|\le q$.
     Furthermore, there are at least $1-\sqrt{\delta}$ fraction of $A\in E(\ma{H}_l^{(t)}(q))$ such that
         $|d_{\mathcal{P}}(A)|\ge (1-\sqrt{\delta})q.$
    \end{claim}

   Let $$\mathcal{R}'=\{A\in \mathcal{R}: |d_{\mathcal{P}}(A)|\ge (1-\sqrt{\delta})q\}.$$ Then, by \cref{claim} and \eqref{Rcardinality}, with high probability we have
    $$|\mathcal{R}'|\ge |\ma{R}|-\sqrt{\delta}|E(\ma{H}_l^{(t)}(q))|\ge(1-2\eta-\sqrt{\delta})|\mathcal{R}|.$$
    As $d_{\mathcal{H}}(A)$ is a sum of independent Bernoulli random variables,  by \cref{prop:ineq} (ii), we have that for each $A\in\mathcal{R}$,
    \begin{equation*}
        \Pr\left[|d_{\mathcal{H}}(A)-\E[d_{\mathcal{H}}(A)]|\le\sqrt{\delta}\E[d_{\mathcal{H}}(A)]\right]=\Pr\left[|d_{\mathcal{H}}(A)-d_{\ma{P}}(A)\cdot p]|\le\sqrt{\delta}d_{\ma{P}}(A)\cdot p\right]\ge1-\sigma,
    \end{equation*}
    where $\sigma \rightarrow 0$ as $q\rightarrow \infty$. As for each $A\in\mathcal{R}'$, $(1-\sqrt{\delta})q\le d_{\ma{P}}(A)\le q$. We have that for each $A\in\mathcal{R}'$, with probability at least $1-\sigma$,
    \begin{equation}\label{eq2}
        |d_{\mathcal{H}}(A)-pq|\le2\sqrt{\delta} pq.
    \end{equation}

    Let $\mathcal{R}''=\{A\in \mathcal{R}':A\text{ satisfies \eqref{eq2}}\}$. Then, according to the above discussion, $$\E[|\mathcal{R}'\setminus\mathcal{R}''|]\le |R'|\cdot \sigma.$$
    By \cref{prop:ineq} (i), we have
    \begin{equation*}
        \Pr[|\mathcal{R}'\setminus\mathcal{R}''|\ge \sqrt{\sigma}|\mathcal{R}'|]\le \frac{\mathbb{E}[|\mathcal{R}'\setminus\mathcal{R}''|]}{\sqrt{\sigma}|\mathcal{R}'|}\le \sqrt{\sigma},
    \end{equation*}


     \noindent which implies that with high probability $|\mathcal{R}'\setminus\mathcal{R}''|\le \sqrt{\sigma}|\mathcal{R}'|\le \sqrt{\sigma}|\mathcal{R}|$. It follows that for all $A\in\mathcal{R}''$ and therefore for all but at most $(2\eta+\sqrt{\delta}+\sqrt{\sigma})|\mathcal{R}|$ choices of $A\in \mathcal{R}$, \eqref{eq2} holds. 
     We conclude that with high probability, the random hypergraph $\ma{H}$ satisfies \cref{lem:pippenger} (i) with $D:=pq$.

We proceed to show that
$\mathcal{H}$ has bounded maximum degree. According to the upper bound in \cref{claim}, for each $A\in \ma{R}$, $d_\ma{H}(A)$ is a sum of at most $q$ independent Bernoulli random variables. Applying \cref{prop:ineq} (ii) with $t=q^{0.8}$ gives that $$\Pr[d_\ma{H}(A)\ge \mathbb{E}[d_\ma{H}(A)]+q^{0.8}]\le e^{-\Theta(q^{0.6})}.$$ By the union bound,
    \begin{equation*}
    \Pr[\exists A\in \ma{R}\text{ such that }d_\ma{H}(A)\ge \mathbb{E}[d_\ma{H}(A)]+q^{0.8}]\le \binom{l}{t}q^{t}e^{-\Theta(q^{0.6})}=o(1).
    \end{equation*}
     It follows that with high probability
     for all $A\in \mathcal{R}$, $d_{\mathcal{H}}(A)\le \mathbb{E}[d_{\mathcal{H}}(A)]+q^{0.8}$,  which implies that $\ma{H}$ satisfies \cref{lem:pippenger} (ii) with high probability.

Lastly, we show that $\mathcal{H}$ has small maximum codegree. Recall that for all distinct $U,U'\in\mathcal{P}$, $|U\cap U'|\le t$. Since for all distinct $A,B\in\mathcal{R}$, $|A\cup B|\ge t+1$, we have
    \begin{equation}\nonumber
    d_{\mathcal{H}}(A,B)=|\{U\in\mathcal{P}:A\subseteq U,B\subseteq U, \binom{U}{t}\cap\mathcal{R} \text{ is a copy of }\mathcal{F}\}|\le 1.
    \end{equation}
    It follows that $\Delta_2(\mathcal{H})\le 1$. Therefore, $\ma{H}$ satisfies \cref{lem:pippenger} (iii).

To conclude, with high probability, the random hypergraph $\mathcal{H}$ constructed above satisfies $|V(\mathcal{H})|=|\mathcal{R}|\ge (1-o(1))\cdot\binom{l}{t}q^t$ and also the assumptions of \cref{lem:pippenger}. Therefore, $\mathcal{H}$ has a matching of size at least $(1-o(1))\cdot\frac{|V(\mathcal{H})|}{|E(\mathcal{F})|}=(1-o(1))\cdot\frac{\binom{l}{t}q^t}{|E(\mathcal{F})|}$.
\end{proof}

It remains to prove \cref{claim}.

\begin{proof}[Proof of \cref{claim}]
    To prove the first part of the claim, note that for distinct $U,U'\in\mathcal{P}$, $|U\cap U'|\le t$. Therefore, for each $A\in E(\ma{H}_l^{(t)}(q))$, the members in $\{U/A:U\in \mathcal{P},A\subseteq U\}$ are pairwise disjoint.
    It follows that $d_\ma{P}(A)\le \frac{(l-t)q}{l-t}=q$.

    To prove the second part of the claim,  assume that there are more than $\sqrt{\delta}$ fraction of $A\in \mathcal{R}$ such that $d_{\mathcal{P}}(A)<(1-\sqrt{\delta})q$. Then
    \begin{equation}\label{eq>}
         \begin{aligned}
            \sum_{A\in E(\ma{H}_l^{(t)}(q))}d_{\ma{P}}(A)
            & =\sum_{\substack{A\in E(\ma{H}_l^{(t)}(q)):\\d_{\mathcal{P}}(A)<(1-\sqrt{\delta})q}}d_{\ma{P}}(A)+\sum_{\substack{A\in E(\ma{H}_l^{(t)}(q)):\\d_{\mathcal{P}}(A)\ge (1-\sqrt{\delta})q}}d_{\ma{P}}(A)\\
            & <\sqrt{\delta}\bi{l}{t}q^t\cdot(1-\sqrt{\delta})q+(1-\sqrt{\delta})\bi{l}{t}q^t\cdot q.
         \end{aligned}
    \end{equation}

    \noindent On the other hand, observe that
    \begin{equation}\label{eq=}
        \sum_{A\in E(\ma{H}_l^{(t)}(q))}d_{\mathcal{P}}(A)=\sum_{U \in \mathcal{P}}\binom{l}{t}=|\mathcal{P}|\binom{l}{t}
    \end{equation}
    and recall that
    \begin{equation}\label{eq<}
            |\mathcal{P}|\ge(1-\delta)q^{t+1}.
    \end{equation}
    Combining \eqref{eq>}, \eqref{eq=}, and \eqref{eq<}, one can deduce that $1-\delta<1-\delta$, a contradiction. Hence, there are at most $\sqrt{\delta}$ fraction of $A\in E(\ma{H}_l^{(t)}(q))$ with $|d_{\mathcal{P}}(A)|< (1-\sqrt{\delta})q$, completing the proof of the claim.
  \end{proof}

\begin{proof}[Proof of \cref{thm:hyper-pac}]
Combining the Lemmas \ref{lem:packing-matching} and \ref{lem:optimal-matching}, \cref{thm:hyper-pac} follows fairly straightforwardly. 
\end{proof}

\section{Concluding remarks}\label{sec:con}

\noindent In this paper we determine, for all fixed integers $c,l$, the exact value of the limit $R_{c,l}=\lim_{q\rightarrow\infty}\fr{M_{c,l}(q)}{q^{\lc l/c \rc}}$. The methods applied by this paper are quite different from those applied by previous papers. For example, the constructions of frameproof codes listed below \eqref{eq:Rcl} are all algebraic and rely on the properties of finite fields. In particular, those constructions are all explicit. However, our construction is probabilistic and therefore non-explicit. The following question remains open.

\begin{question}
    Can we explicitly construct, for general values of $c,l$, frameproof codes with cardinality asymptotically matching the probabilistic construction given in \cref{thm:lower-bd}?
\end{question}

\section*{Acknowledgements}

\noindent This project is supported by the National Key Research and Development Program of China under Grant No. 2021YFA1001000, the National Natural Science Foundation of China under Grant Nos. 12101364 and 12231014, and the Natural Science Foundation of Shandong Province under Grant No. ZR2021QA005. The last author was grateful to Prof. Ying Miao for bringing frameproof codes \cite{blackburn2003frameproof} into his attention in 2014.

{\small
\bibliographystyle{plain}
\bibliography{fpc}
}

\end{document}